\documentclass[11pt, letter]{article}

\usepackage[english]{babel}
\usepackage[T1]{fontenc}
\usepackage{stmaryrd}
\usepackage{algorithmic}
\usepackage{amssymb}
\usepackage{marginnote}
\usepackage[ruled,vlined,linesnumbered]{algorithm2e}
\usepackage{booktabs,tabularx,threeparttable}
\usepackage[normalem]{ulem}

\usepackage{tikz}
\usetikzlibrary{fit,backgrounds,positioning}
\usepackage{fullpage}
\usepackage{mathrsfs}
\usepackage{amsmath}
\usepackage{bbm}
\usepackage{graphicx}
\usepackage{stackengine}
\usepackage{amsthm}
\usepackage{subcaption}
\usepackage{authblk}

\usepackage{thmtools}
\usepackage{thm-restate}

\usepackage[colorinlistoftodos]{todonotes}
\usepackage[colorlinks=true, allcolors=blue]{hyperref}

\usepackage{fancyhdr,lastpage}
\usepackage{amsfonts}
\usepackage[dvipsnames]{xcolor}

\usepackage{enumitem}

\newtheorem{thm}{Theorem}[section]

\newtheorem{obs}[thm]{Observation}

\newtheorem{cl}[thm]{Claim}
\theoremstyle{remark}

\newcommand{\set}[1]{\left\{#1\right\}}

\newcommand{\bb}[0]{\mathbb}

\newcommand{\R}{\bb{R}}

\newcommand{\w}{w}
\usetikzlibrary{shapes.geometric}

\begin{document}

\title{The Popular Dimension of Matchings}

\author[1]{Frank Connor}
\author[1]{Louis-Roy Langevin}
\author[1]{Ndiam\'e Ndiaye}
\author[1]{Agn\`es Totschnig}
\author[2]{Rohit Vasishta}
\author[12]{Adrian Vetta}

\affil[1]{Department of Mathematics and Statistics, McGill University}
\affil[2]{School of Computer Science, McGill University}
\affil[ ]{\texttt {\{frank.connor,louis-roy.langevin,ndiame.ndiaye,agnes.totschnig\}@mail.mcgill.ca}}
\affil[ ]{\texttt {\{rohit.vasishta,adrian.vetta\}@mcgill.ca}}

\date{}

\maketitle

\begin{abstract}
We study popular matchings in three classical settings: the house allocation problem, the marriage problem, and the roommates problem. In the popular matching problem, (a subset of) the vertices in a graph have preference orderings over their potential
matches. A matching is popular if it gets a plurality of votes
in a pairwise election against any other matching.
Unfortunately, popular matchings typically do not exist.
So we study a natural relaxation, namely popular winning sets which are a set of matchings that collectively get a plurality of votes in a pairwise election against any other matching. 
The {\em popular dimension} is the minimum cardinality of a popular winning set, in the worst case over the problem class. 

We prove that the popular dimension is exactly $2$ in the
house allocation problem, even if the voters are weighted and
ties are allowed in their preference lists.
For the marriage problem and the roommates problem, we prove that the popular dimension is between $2$ and $3$,
when the agents are weighted and/or their preferences 
orderings allow ties. In the special case where 
the agents are unweighted and have strict preference orderings, the popular dimension of the marriage problem is known to be exactly $1$ and we prove the popular dimension of the roommates problem is exactly $2$.
\end{abstract}

\section{Introduction}\label{sec:introduction}
A classical problem in matching markets is {\em matching under ordinal preferences}, where each agent has a preference ordering either over a set of items or over the other agents. Three settings are of particular importance:
\begin{enumerate}
\item {\em The House Allocation Problem.} There is a set of house buyers/renters (agents) with preference rankings over a set of houses (items). 
\item {\em The Marriage Problem.} There is a set of men and a set of women where each person (agent) ranks members of the opposite sex.  
\item {\em The Roommates Problem.}
There is a set of students (agents) with preference rankings over each other as potential roommates in a university dorm with two beds per room. 
\end{enumerate}
These problems naturally correspond to finding a matching in
a bipartite graph (the house allocation and marriage problems) or in a non-bipartite graph (the roommate problem). They are of interest as they model a plethora of other applications, such as 
matching employees to tasks, jobs to machines, 
job applicants to employers, 
medical residents to hospitals, or work-partners in a company~\cite{EIV23,Hae18}.

The central line of research concerns finding an equilibrium (or stable solution) in these markets, if one exists. Of course, different concepts of stability are used in the literature. 
Notably, for {\em Pareto stability} an equilibrium requires only the weak requirement that the matching is not Pareto dominated by another matching\footnote{A solution is
{\em Pareto optimal (stable)} if there is no other solution in which every agent is at least as well off and at least one agent is strictly better off.}. In contrast, for {\em core stability} an equilibrium requires the much stronger requirement that the matching is not blocked by a pair of agents\footnote{A solution is in the {\em core} if no coalition of agents can benefit by unilaterally deviating  from the solution. In the case of matching problems the relevant coalitions of interest have cardinality two.}.

In this paper, we study an intermediate stability concept called
{\em popularity} which can be motivated through the lens of ``electing'' a winning matching. Pareto stability states that a proposed election winning matching can be vetoed only if all the agents are unanimous in preferring an alternate matching. On the other hand,
core stability states that a proposed election winner may be vetoed because of just two dissenting agents.
Between these two extremes, popular stability states that a proposed election winner can be vetoed only if all the agents prefer an alternate matching in a plurality vote.
Thus, a {\em popular matching} is a matching that gets a plurality of votes against any other matching in a pairwise election.

Popular matchings were introduced by G\"ardenfors~\cite{Gar75} in 1975 in the context of the marriage problem. They have been studied intensively in computer science and social choice since Abraham et al.~\cite{AIK07} investigated their application in the house allocation problem. Before discussing our results and background literature, let us formally define the popular matching model for our three aforementioned problems.

\subsection{The Popular Matching Model}
In the {\em house allocation problem}, we are given a bipartite graph
$G=(I\cup J,E)$ where $I$ is the set of agents and $J$ is the set of houses (items). Each agent $i$ has a preference order $\succeq_i$
over the set of houses $\Gamma(i)\subseteq J$ it is adjacent to in the graph. In the {\em marriage problem}, we are given a bipartite graph
$G=(M\cup W,E)$ where the set of agents is $M\cup W$, where $M$ and $W$ are respectively a collection of men and women. Each man $m$ has a preference order $\succeq_m$
over the women $\Gamma(m)\subseteq W$ he is adjacent to, and each woman $w$ has a preference order $\succeq_w$ over the men $\Gamma(w)\subseteq M$ she is adjacent to. Finally, in the {\em roommates problem}, we have a non-bipartite graph
$G=(V,E)$ where $V$ is the set of agents (students) and agent $i$ has a preference order $\succeq_i$ over its neighbouring agents $\Gamma(i)\subseteq V$ in $G$. 

We study four cases. First, we distinguish between strict preference lists and weak preference lists where ties are allowed.
Second, we assume each agent $i$ is given a
non-negative weight $w(i)$. In particular, we distinguish between
the unweighted case where $w(i)=1, \forall i$ and the weighted case where $w(i) \ge 0$ is free.

Given two matchings $M$ and $M'$ in the graph, we say that agent $i$ strictly prefers $M$ over $M'$ if (i) it strictly prefers its match in $M$ over its match in $M'$, or (ii) it is matched in $M$ but unmatched in $M'$. We denote this strict preference by $M \succ_i M'$. We then define $\phi(M, M') = \sum_{i: M \succ_i M'} w(i)$ to be the total weight of agents that strictly prefer $M$ over $M'$.
If $\phi(M, M')\ge \phi(M', M)$ then we say $M$ is at least as popular as $M'$.\footnote{We remark that the relation is not transitive.} A matching $M$ is {\em popular} if it is at least as popular as any other matching.

For the marriage problem, a popular matching always exists in
the special case of unweighted agents with strict preference
lists. This is because in this case a stable matching exists~\cite{GS62}
and, furthermore, G\"ardenfors~\cite{Gar75} proved a stable matching is popular. In fact, Huang and Kavitha~\cite{HK13} proved that a stable matching is a popular matching of minimum size.

For the roommate problem, Chung~\cite{Chu00} showed stable matchings are again popular. Unfortunately, stable matchings need not exist in the stable roommates problem~\cite{GS62}; however, there is a
polynomial time algorithm to find a stable matching if one exists~\cite{Irv85}.

In the house allocation problem it is easy to build examples
where a popular matching does not exist~\cite{AIK07}. Indeed non-existence is the norm. Specifically, Mahdian~\cite{Mah06} proved that when the 
agents have uniform independent random preferences there is no popular
matching with high probability, unless the number of houses
is significantly greater than the number of agents; see also~\cite{RI19}.
 


Consequently, an important line of research has been to characterize when a matching is popular.
Abraham et al.~\cite{AIK07} gave a graphical characterization for popular matchings in the house allocation problem.
For the marriage problem, a graphical characterization was given by~\cite{HK13}. These characterizations lead to polynomial time
algorithms to test for the existence of a popular matching and outputs a maximum cardinality popular matching if it exists. 
An optimization-based characterization in terms of maximum weight matchings was given by Bir\'o et al.~\cite{BIM10}. 

The study of agents with weighted voting powers was instigated by Mestre~\cite{Mes14}
for the house allocation problem and by Heeger and Cseh~\cite{HC24} for the marriage problem. See also~\cite{IW10, Kam20,SM10}.

The general non-existence of popular matching led McCuthchen~\cite{McC08} to study multiplicative and additive relaxations called the unpopularity factor and unpopularity margin, respectively. He showed that finding matchings to minimize these
criteria is hard.
An alternate approach around the nonexistence problem
is to consider fractional matchings.
It follows from the seminal works of Kreweras~\cite{Kre65} and Fishburn~\cite{Fis84} 
on maximal lotteries that there exist probability distributions 
over matchings that are at least as popular as any other matching.
Thus, so-called {\em mixed popular matchings}~\cite{HK17,Kav16,KMN11} are guaranteed to exist.
It is the generic nonexistence of popular matchings that also motivates our work. However, we circumvent non-existence via an alternate approach called the {\em popular dimension}, described below.

We remark that numerous extensions to the popular matching problem have been studied, incorporating vertex capacities, edge valuations,  matroid constraints, etc. 
See for example~\cite{Csa24, CKT24, Kam16, Kam17, Kam20, KKM22a, MS06, Pal14, SM10}. 
For more details
, we refer the reader to the survey~\cite{Cseh17} by Cseh.

In parallel with our work, Kavitha et al. studied a closely related setup, focusing on the Condorcet dimension of various generalizations of matchings \cite{KMS26}. While their results were obtained independently, they overlap with some of those presented here.

 

\subsection{The Popular Dimension and Our Results}
Recall the electoral motivation behind popular matchings.
We can exploit this correspondence to
overcome the nonexistence problem for popular matchings.
Specifically, a weak {\em Condorcet winner} is a candidate that is weakly preferred by at least half the voters in a pairwise election against any other candidate.
Thus popular matchings are a stronger notion than Condorcet winners in our election where agents vote for matchings. 
However, a Condorcet winner typically does not exist in an election.
Hence, Elkind et al.~\cite{ELS15} proposed the concept of a {\em Condorcet winning set}, which is a ``committee'' $\mathcal{C}$ of candidates such that no other candidate is strictly preferred by more than half the voters over every candidate in $\mathcal{C}$.
The {\em Condorcet dimension} of an election is the minimum cardinality of a Condorcet winning set. The Condorcet dimension is finite (as the set of all candidates forms a Condorcet winning set) and is one if and only if a Condorcet winner exists.

While the concept of being undefeated is subtly different in the Condorcet winner problem and the popular matching problem (we discuss this in more detail in Section~\ref{sec:disc}), we can follow this approach of committees for popular matchings.
We say that a collection of matchings $\mathcal{M}=\{M_1,\dots,M_k\}$ is a {\em popular winning set} if a plurality of voters prefer $\mathcal{M}$ over any matching $M'$. 
Mathematically, we write $\mathcal{M} \succ_i M'$
if agent $i$ strictly prefers its best match in $\mathcal{M}$ over its match in $M'$. Similarly, 
we say $M' \succ_i \mathcal{M}$
if agent $i$ strictly prefers its match in $M'$ over its best match in $\mathcal{M}$.
 We then define $\phi(\mathcal{M}, M') = \sum_{i: \mathcal{M} \succ_i M'} w(i)$ to be the total weight of agents that strictly prefer $\mathcal{M}$ over $M'$. Similarly, we write 
 $\phi(M', \mathcal{M}) = \sum_{i: M' \succ_i \mathcal{M}} w(i)$.
If $\phi(\mathcal{M}, M')\ge \phi(M', \mathcal{M})$ then we say the set of matchings $\mathcal{M}$ is at least as popular as the matching $M'$. A set of matchings $\mathcal{M}$ is a popular winning set if it is at least as popular as any other matching.
The {\em popular dimension of an instance} is the minimum cardinality $k$ of a popular winning set. The {\em popular dimension} of a class of instances is the maximum over all instances in that class. 

The focus of this paper is to evaluate the popular dimension
of the house allocation, marriage, and roommates problems. 
Our results are summarized in Table~\ref{tab:pop-dim}. Each cell contains the possible values for the dimension of the associated class of instances. In particular, we prove in Section~\ref{sec:housing} that the popular dimension for the house allocation problem is exactly $2$; this result applies for all four cases. 
In Sections~\ref{sec:marriage} and~\ref{sec:roommates}, we study the marriage problem and the roommates problem respectively. In the case of the marriage problem with unweighted agents having strict preferences, the popular dimension is exactly one, as it has been long known that popular matchings always exist. For the roommates problem with unweighted agents having strict preferences, we prove that the popular dimension is exactly $2$. For the three more general cases in both the marriage and roommates problems, we show that the popular dimension is at least $2$ and at most $3$.
\begin{table}[h!]
  \centering
  \begin{threeparttable}
    \caption{Bounds on the Popular Dimension of Matchings}
    \label{tab:pop-dim}
    \begin{tabularx}{\textwidth}{@{}lXXX@{}}
      \toprule
      \textbf{Case} &
      \textbf{House Allocation} &
      \textbf{Marriage} &
      \textbf{Roommates} \\
      \midrule
      Unweighted + Strict      & \textcolor{ForestGreen}{$2 $} (Thm.~\ref{thm:house-general-2}) & 1 \cite{GS62,Gar75} & \textcolor{ForestGreen}{$2 $} (Thm.~\ref{thm:roommates-basic-2}) \\[2pt]
      Weighted + Strict        & \textcolor{ForestGreen}{$2 $} (Thm.~\ref{thm:house-general-2}) & \textcolor{orange}{$2, 3 $} (Thm.~\ref{thm:marriage-general-3}) & \textcolor{orange}{$2, 3 $} (Thm.~\ref{thm:roommates-general-3}) \\[2pt]
      Unweighted with Ties     & \textcolor{ForestGreen}{$2 $} (Thm.~\ref{thm:house-general-2}) & \textcolor{orange}{$2, 3 $} (Thm.~\ref{thm:marriage-general-3}) & \textcolor{orange}{$2, 3 $} (Thm.~\ref{thm:roommates-general-3}) \\[2pt]
      Weighted with Ties       & \textcolor{ForestGreen}{$2 $} (Thm.~\ref{thm:house-general-2}) & \textcolor{orange}{$2, 3 $} (Thm.~\ref{thm:marriage-general-3})   & \textcolor{orange}{$2, 3 $} (Thm.~\ref{thm:roommates-general-3}) \\
      \bottomrule
    \end{tabularx}
    \begin{tablenotes}[para,flushleft]
      \footnotesize
      \textcolor{ForestGreen}{$\textbf{Green}$} = The exact dimension is now known \hspace{8mm}
      \textcolor{orange}{$\textbf{Orange} $} = The exact dimension remains open     
    \end{tablenotes}
  \end{threeparttable}
\end{table}

\subsection{Discussion}\label{sec:disc}
Huang and Kavitha~\cite{HK17} proved that, with unweighted agents and strict preferences, the polytope of popular fractional matchings (that is, mixed popular matchings) is half-integral in the roommates problem, and thus also in the marriage problem. In the marriage problem this half-integrality yields a mixed matching supported on two integral matchings, hence a popular dimension of at most $2$. Although, as discussed with respect to Table~\ref{tab:pop-dim}, in this setting the popular dimension was already known to be exactly $1$. 
For the roommates problem while half-integrality holds, the resultant fractional matchings need not decompose into two integral matchings, but it can be shown that they decompose into supports of cardinality three. This shows the popular dimension is at most $3$; we improve this by proving the popular dimension is exactly $2$ in this case. In addition, Huang and Kavitha~\cite{HK17} show that half-integrality fails once ties are allowed in the preference lists for the  
marriage and roommates problem.
This, in turn, implies that half-integrality does not hold for the house allocation problem even with unweighted agents and strict preferences.








It is worth clarifying the relationship between 
the popular dimension and the Condorcet dimension.
The popular dimension is lower bounded by the Condorcet dimension. To see this, we remark that there is a subtle but important distinction
between popular winning sets and Condorcet winning sets.
Voters who are indifferent are assumed to vote for the proposed Condorcet winning sets.
That is, for $\mathcal{M}$ to be defeated by an alternate
matching $M'$ it must be the case that more than half the voters (a majority) strictly prefer $M'$ over $\mathcal{M}$. In particular, even if a plurality of voters strictly prefer $M'$ over $\mathcal{M}$, the committee $\mathcal{M}$ still wins if that plurality does not form a majority.

In contrast, for popular winning set, the votes of voters who are indifferent between the proposed popular winning set $\mathcal{M}$ and the alternative $M'$ are not counted (or, equivalently, half these votes are assigned to both sides). Thus $\mathcal{M}$ is defeated by $M'$ if and only if a strict plurality of the voters prefer $M'$ over $\mathcal{M}$.
It immediately follows that a popular winning set is 
a Condorcet winning set, but the converse is not always the case. 

Hence, the popular dimension may be higher than the Condorcet dimension. Elkind et al.~\cite{ELS15} proved that the Condorcet dimension is at most $\lceil \log m \rceil$ in any election, where $m$ is the number of candidates. Lassota et al.~\cite{LVV24} conjectured that the Condorcet dimension is
at most $3$ for any election, and proved the conjecture
holds in special cases of the spatial electoral model in two dimensions. Recently, Charikar et al.~\cite{CLR25} proved the
Condorcet dimension is at most $6$ in any election.

The Condorcet winning set given by the algorithm of~\cite{ELS15} does in fact satisfy the stronger popularity condition. But since the number of matchings (candidates) is exponential, this only implies an upper bound of $O(n\log n)$ on the popular dimension.
The Condorcet winning set of cardinality output in~\cite{CLR25}
does not necessarily satisfy the popularity condition.
So their upper bound of $6$ does not immediately extend to the popular dimension.
On the other hand, our upper bounds of at most $3$ on the popular dimension imply an upper bound on the Condorcet
dimension. Thus our results prove that the conjecture of~\cite{LVV24} holds for the house allocation, marriage, and roommates problems.

\section{The House Allocation Problem}\label{sec:housing}
First we consider the house allocation problem.
Specifically, we prove the popular dimension is exactly $2$ for all four cases of the house allocation problem.
The lower bound is given in Section~\ref{sec:house-lower}
and the corresponding upper bound is shown in
Section~\ref{sec:house-upper}.

\subsection{A Lower Bound on the Popular Dimension}\label{sec:house-lower}

It is well-known that popular matchings need not exist
in the housing allocation problem~\cite{AIK07}. We include a proof of this fact for completeness.
\begin{obs}[Abraham et al. \cite{AIK07}]\label{obs:house-lower-bounds}
There exist instances of the house allocation problem
with popular dimension at least $2$, even for the case of unweighted agents with strict preference orderings.
\end{obs}
\begin{proof}
Consider the instance consisting of a complete bipartite graph $(I,J)$ with agents $I = \{a,b,c\}$ and houses $J = \{x,y\}$. The agents
all strictly prefer house $x$ over house $y$. Let $M$ be any matching, without loss of generality $c$ is unmatched and $b$ is not matched to $x$ in $M$. Then, consider the alternate matching $M' = \{bx,cy\}$, then both $b$ and $c$ are happier with $M'$ than with $M$, which implies $M$ is not popular. Since $M$ is an arbitrary matching, no matching is popular, hence the dimension of this instance is at least 2.
\end{proof}

\subsection{An Upper Bound on the Popular Dimension}\label{sec:house-upper}
The task now is to prove that the popular dimension of the house
allocation problem is at most $2$, even for the general
case of weighted agents with preference lists that may contain ties. First we introduce some notation. 
Given a set $I'\subseteq I$ of agents and a set $J'\subseteq J$ of houses, we define a \textit{$(1,2)$-matching} from $I'$ to $J'$ to be a graph on $(I',J')$ such that each agent in $I'$ is adjacent to a house in $J'$ and no house in $J'$ has more than 2 agents adjacent to it. We call it a \textit{top-choice $(1,2)$-matching} if each agent in $I'$ (weakly) prefers their match over every other house of $J'$. Furthermore, for an agent $i$ and a set $J' \subseteq J$, define $f_i(J') = \set{j \in J': j \succeq_i j' \enspace\forall j' \in J'}$, which is the set of the favourite houses of agent $i$ among houses in $J'$. In particular, for a top-choice $(1,2)$-matching, the edge incident to any $i \in I$ must be of the form $(i,j)$ for some $j \in f_i(J)$.

Our interest in $(1,2)$-matchings is because they can always be decomposed in polynomial time into the union of two matchings $M_1\cup M_2$. 
Specifically, top-choice $(1,2)$-matchings will help us generate a popular winning set
$\mathcal{M}=\{M_1,M_2\}$.
Indeed, Algorithm~\ref{AI:General}
outputs a popular winning set of cardinality $2$.


\begin{algorithm}[h!]
\textbf{Input:} A set $I$ of agents, weights $w: I \to \R^+$, preferences $(\succeq_i)_{i \in I}$ over a set of houses $J$.

    Let $I_1=I$, $J_1 = J$, $E_0=\emptyset$ and $t=0$.\\
    \While{$I_{t+1},J_{t+1} \ne \emptyset$}{
        Increment $t$ by $1$.\\
        
        Relabel the agents in $I_t$ as $i_1,\ldots,i_{|I_t|}$ such that $w(i_1) \ge \cdots \ge w(i_{|I_t|})$.\\ 
        Let $k_t$ be maximal s.t. $\exists $ a top-choice $(1,2)$-matching from $\{i_1,\ldots,i_{k_t}\}$ to $J_t$.\\
        \If{$k_t = |I_t|$}{
            Find a top-choice $(1,2)$-matching $M^*_t$ from $I_t$ to $J_t$ and set $\hat{I}_t=I^*_t=I_t$.
        }
        \Else{
            Let $\hat{I}_t=I'_t= \{i_1,\ldots,i_{k_t}\}$ and set  $W_t=w(i_{k_t+1})$. \\
            \For{$\ell = k_t$ down to $1$}{
                \If{$\nexists$ top-choice $(1,2)$-matching from $I'_t\cup \{i_{k_t+1}\}\setminus\{i_\ell\}$ to $J_t$}{
                    Set $I'_t \leftarrow I'_t\setminus\{i_\ell\}$.
                }
            }
            Set $I^*_t \leftarrow I'_t$ and find a top-choice $(1,2)$-matching $M^*_t$ from $I^*_t$ to $J_t$.\\
        }
        Set $E_t\leftarrow E_{t-1}\cup M^*_t$.\\
        Set $J^*_t$ to be the set of degree 2 houses matched in $M_t^*$.\\
        Set $I_{t+1}\leftarrow I_t\setminus I^*_t$.\\
        Set $J_{t+1}\leftarrow J_t\setminus J^*_t$\\
    }
    Let $T\leftarrow t$\\
    There must exist two matchings $M_1$ and $M_2$ s.t. $M_1\cup M_2=E_T$. Output $\mathcal{M} = \set{M_1,M_2}$.
    \caption{Size 2 Popular Winning Set in the general house allocation problem}
    \label{AI:General}
\end{algorithm}

Algorithm~\ref{AI:General} works in time steps.
Informally, it orders the agents by weight and in each time step $t$ it finds a maximal top-choice $(1,2)$-matching among the highest weight agents (lines 5-6). Unfortunately this matching does not possess the properties we require to generate popular winning sets.
In particular, we want to make sure that agents who are allocated a house at time $t$ strictly prefer the house they are allocated over all houses that have yet to be allocated. This necessitates removing $i$ and possibly numerous other agents from the set of agents matched at this time step
(for exact details of how this is done, see lines 10-14). This will help us obtain a top-choice $(1,2)$-matching $M_t^*$ satisfying the property. We then remove the agents and houses in $M_t^*$ and recurse on the remaining agents and houses.
The disjoint union of the top-choice $(1,2)$-matchings found in each time step is itself a $(1,2)$-matching and thus decomposes into two matchings $\mathcal{M}=\{M_1, M_2\}$. We remark that each agent $i$ gets at most one house in $M_1$ or $M_2$, so, when convenient, we will refer to that one house as the house that $i$ gets in $\mathcal{M}$, which we denote $\mathcal{M}(i)$.


Our task now is to prove that this output $\mathcal{M}$ is a popular winning set.
In order to prove this, we require several claims.
Suppose agent~$i$ is allocated a house $h=\mathcal{M}(i)$ at time $t$. 



The first claim shows that $i$ strictly prefers $h$ over every house that is allocated after time $t$.

\begin{cl}\label{cl:after}
    If $i\in I^*_{t}$, then agent~$i$ strictly prefers its own house $\mathcal{M}(i)$ over any house in $J_{t+1}$.
\end{cl}

\begin{proof}
    Assume, for the sake of contradiction, that $i\in I^*_t$ and that there exists a house $j$ in $J_{t+1}$ such that $i$ does not strictly prefer $h = \mathcal{M}(i)$ over $j$. By the construction of $I^*_t$, if agent $i$ was not removed from $I'_t$ in the lines 11-13 loop of Algorithm~\ref{AI:General}, then $I^*_t\cup \{i_{k_t+1}\}\setminus \{i\}$ must have a top-choice $(1,2)$-matching $M$ over $J_t$ but $I^*_t\cup \{i_{k_t+1}\}$ does not.

    Consider an auxiliary multigraph $H$ over the vertex set $V=(I^*_t\cup \{i_{k_t+1}\})\cup (J_t^{(1)}\cup J_t^{(2)})$, where $J_t^{(1)}$ and $J_t^{(2)}$ are copies of houses of $J_t$. 
    For each $e\in M^*_t\cap M_1$, $H$ has the corresponding edge $e$ from $I^*_t$ to $J_t^{(1)}$ coloured blue. Similarly, we add edges of $M^*_t\cap M_2$ as edges from $I^*_t$ to $J_t^{(2)}$ also coloured as blue. 
    Finally, add the edge $(i,j)$ and the edges of $M$ in a way such that no two edges of $M$ are adjacent to the same house, and color them red. Clearly, such a way exists since $M$ is a $(1,2)$-matching and we have two copies of each house.

    Also, $M$ is a $(1,2)$-matching from $I^*_t\cup \{i_{k_t+1}\}\setminus \{i\}$ to $J_t$, so by adding the edge from $i$ to $j$, each agent in $I^*_t\cup \{i_{k_t+1}\}$ is incident to exactly one red edge. Moreover, by definition, $i_{k_t+1}$ is not matched at step $t$ in $E_t$, so it has no blue edge incident to it in $H$, while all agents of $I^*_t$ are incident to exactly one red edge and one blue edge. 
    
    In particular, agents from $I^*_t$ will either be matched to the same house, that is,  with both a blue and a red edge or they may be matched to two distinct houses. By assumption, $i$ values $j$ at least as much as the house $h$ they are allocated in $\mathcal{M}$ and since $M$ and $M^*_t$ are both top-choice $(1,2)$-matchings, all remaining agents must also value both houses they are allocated equally. We can conclude that all agents with degree $2$ in $H$ value both their edges equally, which implies that a $(1,2)$-matching using only red edges and blue edges is also a top-choice $(1,2)$-matching.
    
    Now $i_{k_t+1}$ is the only agent whose degree is $1$. So there must exist an alternating red-blue path starting at $i_{k_t+1}$ ending at some house of degree 1. By replacing the blue matching edges on the path by the red edges in $M$, we obtain a top-choice matching from $I^*_t\cup \{i_{k_t+1}\}$ into $J^*_t$, a contradiction.

    Ergo, an agent $i\in I^*_t$ cannot weakly prefer $j\in J_{t+1}$ over the house $h$ they get in $\mathcal{M}$.
\end{proof}

Using the previous claim, the next claim shows that at each step $t<T$, the number of agents that get allocated a house is exactly twice the number of houses that get allocated.

\begin{cl}\label{cl:deg}
    For all $t<T$, we have $|I^*_t|=2|J^*_t|$.
\end{cl}

\begin{proof}
    Assume there is a house $j$ that has degree $1$ in $M^*_t$. Then, there must exist some agent $i\in I^*_{t'}$ for $t'\leq t$ such that $\mathcal{M}(i)=j$. However, we get $j\in J_{t+1}$, which contradicts Claim~\ref{cl:after}. Since the degree of all houses in $E_t$ is either 0 or 2 and the degree of all agents is either 0 or 1 for any $t$, clearly the number of agents of degree 1 in $E_t$ is twice the number of houses of degree 2 so $|I\setminus I_{t+1}|=2 |J\setminus J_{t+1}|$. The same can be used to show $|I\setminus I_{t}|=2 |J\setminus J_{t}|$. So, using our definitions from lines 16-18, we get $|I^*_t|=2|J^*_t|$.
\end{proof}

The final claim states that if an agent~$i$ is in $\hat{I}_t$
then their match in $\mathcal{M}$ is a favourite house in $J_t$. Note this is true even if agent $i$ is not allocated a house in step $t$. That is, even if 
$i\in \hat{I}_t\setminus I^*_t$, because it is discarded in the lines 11-13 loop.  
\begin{cl}\label{cl:hat}
    If $i\in \hat{I}_t$, then $i$ gets a house that is their favourite in $J_t$.
\end{cl}
\begin{proof}
    We prove this property by induction on $T-t$.
    
    \textbf{Base Case:} If $T-t=0$, then, by our termination condition for the outer while loop (line 3) we must have that either $I_{T+1}$ or $J_{T+1}$ is empty.   
    
    Assume that $I_{T+1}=\emptyset$, then $\hat{I}_t\setminus I^*_t=\emptyset$. 
    Then, $i\in I^*_t$ so agent $i$ gets a favourite house in $J_t$. 
    
    So assume $J_{T+1}=\emptyset$. In this case $|I^*_t|=2|J_t|$ by Claim~\ref{cl:deg}. So if $\hat{I}_t\neq I^*_t$, we must have $\hat{I}_t>2|J_t|$. By the definition of $k_t$ and $\hat{I}_t$ in lines 6 and 10, there exists a $(1,2)$-matching from $\hat{I}_t$ to $J_t$. However, this is not possible by the pigeonhole principle since $\hat{I}_t>2|J_t|$. So, our assumption was false and we must have $I^*_t=\hat{I}_t$. So $i\in I^*_t$ and $i$ gets a favourite house in $J_t$.
    
    \textbf{Induction Step:} Assume that each agent in $\hat{I}_{t+1}$ gets a house that is its favourite in $J_{t+1}$. We wish to prove that every agent in $\hat{I}_t$ gets its favourite house in $J_t$.
    
    If $i\in I^*_t$ then $i$ gets its favourite house in $J_t$ by construction.
    So assume that $i\in  \hat{I}_t\setminus I^*_t$.  It suffices to show that $i\in \hat{I}_{t+1}$ and $f_{J_t}(i)\cap f_{J_{t+1}}(i)\neq \emptyset$. That is, the best option for $i$ in $J_{t+1}$ is at least as good as its best option in $J_t$.

    By definition of $\hat{I}_t$, there exists a top-choice $(1,2)$-matching $M$ from $\hat{I}_t$ to $J_t$  . By Claim~\ref{cl:after}, no agent in $I^*_t$ has a top-choice from $J_t$ in $J_{t+1}$. 
    In particular, $I^*_t$ must be matched to $J^*_t$ in $M$. By the pigeonhole principle and claim \ref{cl:deg}, this implies that $\hat{I}_t\setminus I^*_t$ must be matched outside of $J^*_t$ in $M$. So, even after step $t$, $\hat{I}_t\setminus I^*_t$ has a top choice $(1,2)$-matching in $J_{t+1}$, as their match in $M$, which is a favourite house in $J_t$, remain favourite and available in $J_{t+1}$. Therefore $\hat{I}_t\setminus I^*_t\subseteq \hat{I}_{t+1}$ as these agents have priority in the step $t+1$ due to the weight ordering. The claim follows by induction on $T-t$.
\end{proof}

We can now use our claims to prove the main result of this section.

\begin{thm}\label{thm:house-general-2}
 A popular winning set of size $2$ can be found in polynomial time for the general case of weighted agents with preference ties allowed. In particular, the popular dimension is at most $2$.
\end{thm} 
\begin{proof}
    Let $M'$ be any alternate matching on $(I,J)$. Let $i$ be an agent who strictly prefers their match in $M'$ over their match $h = \mathcal{M}(i)$.
    Let $\hat{t}_i$ be the minimum $t$ such that $i\in \hat{I}_t$. For each $j\in J\setminus J_{T+1}$, let $t^*_j$ be the unique index such that $j\in J^*_{t^*_j}$. Then, if $i$ strictly prefers $j$ over $h$, $t^*_j<\hat{t}_i$ by Claim~\ref{cl:hat}. 
    In particular, since $i\notin \hat{I}_{t^*_j}$, we get $w(i)\leq W_{t^*_j}$, where $W_{t^*_j}$ is as defined in line 10 of Algorithm~\ref{AI:General}. Therefore we have:
     \begin{equation}\label{eq:up}
     \phi(M', \mathcal{M}) 
     \leq \sum_{(i,j)\in M'\text{ and }t^*_j<\hat{t}_i} w(i)
     \leq \sum_{j:t^*_j<T} W_{t_j}
     =\sum_{\tau=1}^{T-1} |J_\tau|\cdot W_\tau
      \end{equation}

    Next we wish to find a set of agents who are guaranteed to be worse off in $M'$ compared to $\mathcal{M}$. In order to do so, we define a partition of a subset of agents $S^*_1\dots,S^*_{T-1}$ such that $S_t^* \subseteq I\setminus I_{t+1}$ and $M'(i)\notin J^*_1\cup\dots\cup J^*_t$ whenever $i \in S^*_t$. By Claim~\ref{cl:after}, this implies that $i$ strictly prefers $\mathcal{M}$ over $M'$. We also wish for our partition to contain enough agents from $I^*_1,\dots,I^*_t$ to give a good lower bound on the sum of the weight of the agents who prefer $\mathcal{M}$.
    
    Formally, define sets of agents $S^*_t\subseteq I$, for $t=1,\dots,T-1$, such that:
    \begin{itemize}
        \item $\set{S^*_{1},\dots,S^*_{T-1}}$ are pairwise disjoint.
        \item $\forall t$, $S^*_t\subseteq \bigcup_{\tau=1}^t I^*_\tau \setminus \bigcup_{\tau=1}^t M'(J^*_\tau)$.
        \item $\forall t$, $|S^*_t|=|J^*_t|$ 
    \end{itemize}
    We will show the construction of this partition by induction.

    \textbf{Base Case:} Since $|I^*_1|=2|J^*_1|$ by Claim~\ref{cl:deg}, $I^*_1 \setminus M'(J^*_1)$ has size at least $|J^*_1|$. Let $S^*_t$ be any subset of $I^*_1 \setminus M'(J^*_1)$ of cardinality $|J^*_1|$.
    
    \textbf{Induction Step:}
    Assume that we have selected $S^*_1,\dots,S^*_{t-1}$ satisfying the properties. Since for $\tau=1,\dots, t$ we have $|I^*_\tau|=2|J^*_\tau|$
    by Claim~\ref{cl:deg}, it follows that 
    $\bigcup_{\tau=1}^t I^*_\tau \setminus \left(\bigcup_{\tau=1}^t M'(J^*_\tau)\cup \bigcup_{\tau=1}^{t-1} S^*_\tau\right)$ has cardinality at least $\sum_{\tau=1}^t |I^*_\tau| - \sum_{\tau=1}^t |J_\tau^*| - \sum_{\tau=1}^{t-1} |S_\tau^*|$, which simplifies to $|J^*_t|$. So we can select $S^*_t$ properly, which is what we wanted to show.

    Now consider an agent $i$ in $S^*_t$. By Claim~\ref{cl:after}, since $S^*_t\subseteq \bigcup_{\tau=1}^t I^*_\tau$, agent $i$ is strictly worse off if it gets a house that is not in $\bigcup_{\tau=1}^t J^*_t$. However, by definition of $S^*_t$, $i$ does not get such an house in $M'$, so $i$ strictly prefers the house they are matched to in $\mathcal{M}$. Moreover, we have that $w(i)\geq W_t$. It follows that:
    \begin{equation}\label{eq:low}
    \phi(\mathcal{M},M') 
    \geq 
    \sum_{i\in \bigcup_{\tau=1}^{T-1} S^*_\tau} \w(i)= \sum_{\tau=1}^{T-1} \sum_{i\in S^*_\tau} w(i)\geq \sum_{\tau=1}^{T-1} |J_\tau|\cdot W_\tau
    \end{equation}
  Together (\ref{eq:up}) and (\ref{eq:low}) imply that $\mathcal{M}=\{M_1, M_2\}$ is a popular winning set of cardinality $2$. 
\end{proof}

\section{The Marriage Problem}\label{sec:marriage}
For the marriage problem with unweighted agents with strict preference orderings a stable matching always exists~\cite{GS62}. Moreover, for this case, G\"ardenfors~\cite{Gar75} proved that a stable matching is a popular matching. It immediately follows that the popular dimension is exactly one.
\begin{thm}\label{thm:marriage-basic-1}\textnormal{(Gale and Shapley~\cite{GS62} and G\"ardenfors~\cite{Gar75})}
    The popular dimension of the marriage problem is $1$,
    for the case of unweighted agents with strict preference orderings.
\end{thm}
However, for more general cases of the marriage problem, the popular dimension is at least $2$.
\begin{obs}\label{obs:marriage-lower-bounds}
The popular dimension of the marriage problem is at least $2$, if the agents are 
either weighted or have preference orderings with ties allowed.
\end{obs}
\begin{proof}
Take an arbitrary unweighted house allocation instance with strict preferences $(I,J,E,\succ)$. We show this can be modeled as a marriage problem either (i) using weights or (ii) using ties. 

First, to transform this instance into a weighted marriage problem we let the set of men 
be $M = I$, the set of women be $W = J$, and the set of edges be $E$.
We assign a weight of $1$ to each man $i\in M$ and a weight of $0$ to each woman $j \in W$. Since the popularity condition relies on weighted voting, this means only the votes of the men contribute. Consequently, a matching is popular in this marriage problem if and only if it is popular among the agents $I$ in the original house allocation problem.

Second, to transform this instance into a marriage problem with ties we do the following. Create an intermediary house allocation instance by adding $|I|$ dummy houses $J'$ and $|J|$ dummy agents $I'$. The edge set of the new instance is the union of $E$ (the edges of the original instance), the set of all pairs containing a dummy agent and the set of all pairs containing a dummy house.  An agent $i\in I$ will be indifferent to all dummy houses, strictly prefer being matched to a house in $\Gamma_E(i)$ over being matched to a dummy house and will have the same preferences over $\Gamma_E(i)$ as in the original instance. On the other hand, the dummy agents in $I'$ are indifferent to all houses in $J \cup J'$. Note that the popular dimension of the original and intermediary instances are the same.

 Build an instance of the marriage problem as $M = I \cup I'$, $W = J \cup J'$, edges $E$ and men's preferences as in the intermediary instance, and women are indifferent to all men.
 Since $|I|=|J'|$ and $|I'|=|J|$, by using the pigeonhole principle it is possible to show that any matching in which a man $i \in I$ is unmatched is pareto dominated by one in which $i$ is matched. In particular, since a popular matching cannot be pareto dominated, we may assume all matchings we work with are perfect matchings. Since the condition for popularity requires strict preferences neither the dummy agents nor the women in $W$ affect the popularity of any matching as they are indifferent to all their matches.
Thus, again, only the preferences of the men affect the popularity condition. So a matching is popular in this marriage problem if and only if it is popular among the agents $I$ in the intermediary house allocation instance.

Consequently, the example from Observation~\ref{obs:house-lower-bounds} provides a valid lower bound for the marriage problem either with weighted agents or with ties in the preference orderings. 
\end{proof}

On the positive side, for the marriage problem the popular dimension
is upper bounded by $3$.
\begin{thm}\label{thm:marriage-general-3}
 A popular winning set of size $3$ can be found in polynomial time for the general case of weighted agents with preference ties allowed. In particular, the popular dimension is at most $3$.
\end{thm}
The validity of Theorem~\ref{thm:marriage-general-3} follows from 
a more general result for the roommates problem, namely Theorem~\ref{thm:roommates-general-3}, that we prove in the next section.

\section{The Roommates Problem}\label{sec:roommates}

Lastly, we study the roommates problem.
Unlike the marriage problem, the popular dimension is at least $2$, even for the simple case of unweighted agents with strict preference orderings. 
\begin{obs}\label{obs:roommates-lower-bounds}
The popular dimension of the roommates problem is at least $2$, even for the case of unweighted agents with strict preference orderings. 
\end{obs}
\begin{proof}
Consider an instance of the roommate problem with agents $A = \set{a,b,c}$. The agents are unweighted and have the following strict preferences: $b \succ_a c, c \succ_b, a,$ and $a \succ_c b$. Take, without loss of generality, the matching to be the edge $(a,b)$. Take any non-empty matching, it consists of a single edge. Without loss of generality, this is the edge $(a,b)$. Then the edge $(b,c)$ is strictly preferred by agents $b$ and $c$. Moreover, clearly any empty matching is beaten by any non-empty matching. Thus, no matching can be popular, and the popular dimension is at least $2$. 

We remark that in the roommates problem it is often assumed that the number of agents is even.
Regardless, our lower bound on the popular dimension remains at least $2$. 
To see this take two copies of the above example where agents in one copy do not wish to be matched to agents
in the other copy. 
\end{proof}
We now show in Section~\ref{sec:roommates-u-t} that the
bound of $2$ is tight for the case of unweighted agents with strict preference orderings.
Then, in Section~\ref{sec:roommates-general}, we prove an upper bound of $3$ on the popular dimension for the general case of weighted agents with preference ties allowed.

\subsection{The Roommates Problem with Unweighted Agents and Strict Preferences}\label{sec:roommates-u-t}

We wish to show that any instance of the roommates problem with unweighted agents and strict preferences admits a popular winning set of cardinality 2.

We will prove this by presenting an algorithm. The way the algorithm works is by starting from any agent $v_1$ and greedily building a sequence $v_1,v_2,\ldots,v_n$ such that $v_{i+1}$ is the top-choice of $v_i$ among $V_i=\{v_{i+1},\ldots,v_n\}$ if such a top-choice exists. If no such top choice exists, instead we pick an agent that has not been selected yet.

\begin{algorithm}
    Let $M_1 = M_2 = \emptyset$. Let $v_1$ be any vertex of $V$ and set $t=1$.

    Let $V_t$ be the set of vertices in $V\setminus \set{v_1}$ that are unmatched at time $t$, hence $V_1 = V \setminus\{v_1\}$.
    
    \While{$V_t \ne \emptyset$}{
        \uIf{$v_t$ has a neighbour in $V_t$}{
            Pick $v_{t+1}$ to be the favourite agent of $v_t$ in $V_t$.\\
            \uIf{$t$ is odd}{
                Add $(v_t,v_{t+1})$ to $M_1$.
            }
            \uElse{
                Add $(v_t,v_{t+1})$ to $M_2$.
            }
        }
        \Else{
            Pick $v_{t+1}$ to be any agent in $V_t$.
        }
        Set $V_{t+1} = V_t\setminus\{v_{t+1}\}$.\\
        Increment $t$ by $1$.
    }
    Output $\mathcal{M} = \set{M_1,M_2}$.
\caption{Size 2 Popular Winning Set in the unweighted roommates problem with strict preferences} \label{AA:main}
\end{algorithm}

\begin{thm}\label{thm:roommates-basic-2}
    A popular winning set of size $2$ can be found in polynomial time for the case of unweighted agents with strict preferences. In particular, the popular dimension is $2$.
\end{thm}
\begin{proof}
It is easy to see that Algorithm~$\ref{AA:main}$ terminates in polynomial time.
Furthermore, by construction $M_1$ and $M_2$ are matchings.
Indeed, together $M_1$ and $M_2$ form alternating edges on disjoint paths.
Thus, it only remains to prove that $\mathcal{M} = \{M_1,M_2\}$ is a popular winning set.
Take any matching $M'$. 
Let $v_i$ be an agent that strictly prefers its neighbour $v_j$
in $M'$ over $v_{i-1}$ and $v_{i+1}$ (if they exist).
Note that $v_j\notin V_i = \{v_{i+1},\ldots,v_n\}$, otherwise $v_i$ could not strictly prefer $v_j$ over $v_{i+1}$ by observing line 5.
It follows that $j\le i-2$, but then $i\ge j+2$, so $v_i\in S_{j+1}=\{v_{j+2},\dots, v_n\}$.
Line 5 gives us that $v_j$ strictly prefers $v_{j+1}$ over $v_i$ since $v_i$ and $v_j$ are neighbours.

Consequently, if an agent $u$ strictly prefers $M'$ over
$\mathcal{M}$, then its neighbour in $M'$ strictly prefers $\mathcal{M}$ over $M'$.
But $M'$ is a matching, so the number of agents who are strictly happier with $M'$ is at most the number of agents who are strictly happier with $\mathcal{M}$.
We conclude that $\mathcal{M}$ is at least as popular as
$M'$. Thus we have a popular winning set of cardinality two,
as desired.
\end{proof}

\subsection{The Roommates Problem with Weighted Agents and Ties Allowed}\label{sec:roommates-general}


To prove our bound of 3 on the general case of the roommates problem, we present Algorithm~\ref{AA:w+t} that outputs three matchings $M_1, M_2$ and $M_3$ that together form a popular winning set.

To find a popular winning set of cardinality $3$, the algorithm converts the roommates instance into an instance of the house allocation problem on an auxiliary graph. Running Algorithm~\ref{AI:General} outputs two matchings on the auxiliary graph.
In turn, after transforming back to the original roommates instance, these two matchings induce the three desired matchings. Let's now
formalize this algorithm.

\begin{algorithm}
    Let $V$ be the set of agents. \\

    Set $V^-$ be copies of $V$ treated as houses with no preferences.\\
    Let $V^+$ be copies of $V$ with preferences over $V^-$ matching those of $V$.\\
    Run Algorithm~\ref{AI:General} over $(V^+,V^-)$ to get a set of edges $E'=M'_1\cup M'_2$.\\
    Merge the vertices of $V^+$ with their copy in $V^-$. Denote $E$ the resulting set of edges.\\
    3-colour the resulting set of edges $E$ to obtain 3 matchings $M_1,M_2,M_3$.\\
    Output $\mathcal{M}=\{M_1,M_2,M_3\}$.
\caption{Size 3 Popular Winning Set in the general roommates problem}\label{AA:w+t}
\end{algorithm}

\begin{thm}\label{thm:roommates-general-3}
 A popular winning set of size $3$ can be found in polynomial time for the general case of weighted agents with preference ties allowed. In particular, the popular dimension is at most $3$.
\end{thm}
\begin{proof}

For every line except line 6, it is straightforward to see that they run in polynomial time.
So it remains to prove that line 6 can be run in polynomial time and that the output is a
popular winning set. Recall that Algorithm~\ref{AI:General} actually outputs a popular winning set of size two $\{M'_1,M'_2\}$ where $E' = M'_1\cup M'_2$ is a $(1,2)$-matching $E'$.
By Theorem~\ref{thm:house-general-2}, the pair of matchings $M'_1$ and $M'_2$ form a popular winning set in the auxiliary house allocation problem
on $(V^+,V^-)$. That is, there is no matching $M'$ on $(V^+,V^-)$ that is more popular to the agents of $V^+$ over $E'$. 

Observe that any matching over $V$ naturally corresponds to a matching over $(V^+,V^-)$. Moreover, $E$ and $E'$ correspond to the same set of edges. So assume $E$ can be partitioned into $k$ matchings, $\mathcal{M}=\{M_1, M_2,\dots, M_k\}$.
Then, since the set of agents is identical, $\mathcal{M}$ will be at least as popular as any other matching in the original roommates problem.

Thus it suffices to prove that $E$ is 3-edge colourable in polynomial time
and so decomposes in $k=3$ matchings, that is, $\mathcal{M}=\{M_1, M_2, M_3\}$.
Because Algorithm~\ref{AI:General} outputs a $(1,2)$-matching $E'$, each agent has degree at most one and each house has degree at most two, with respect to the edges in $E'$.
It follows that, with respect to the edges in $E$, each vertex (agent)
has only one edge incident to it that came from its copy in $V^+$ and at most two edges incident to it that came from its copy in $V^-$. 

In particular, let's orient the edges of $E$ as follows.
If $(i, j)\in E'$ with $i\in V^+$ and $j\in V^-$, then we orient $(i,j)$
from $i$ to $j$ in $E$. It follows that, with respect to the edges in $E$, each vertex (agent) in $V$ has out-degree at most $1$ and in-degree at most $2$. Thus $E$ induces a directed graph $D=(V, E)$ with maximum out-degree equal to $1$. This implies that each connected component of $D$ has at most one cycle. 
Therefore, removing the arc orientations, we have that each component of $E$ is a (subgraph of) a tree plus one edge with maximum degree at most $3$.
We then have a polynomial time greedy algorithm to $3$-colour the edges of $E$. To begin we $3$-edge colour the edges of the unique cycle $C$ of each component (if it exists); note that $C$ may be an odd cycle.
Contracting the edges of $C$ gives a tree $T$, which we root at $C$.
Then we colour the remaining edges of $E$ in increasing order of distance
from the root $C$. Suppose we attempt to colour an edge $e=(i,j)\in E$,
where $j$ is closer to the root than $i$. Then,
because the maximum degree is $3$ (and there are no cycles in $T$), 
at most two edges incident to $j$ have been coloured (and none incident to $i$). Thus at least one of the three colours is available to colour $e$. This process gives a 3-edge colouring of $E$ which gives our
three matchings $\mathcal{M}=\{M_1,M_2,M_3\}$.
\end{proof}



\section{Conclusion}
    Motivated by the nonexistence of popular matchings, we introduced the notion of popular dimension. For all cases of the house allocation problem we showed the popular dimension to be exactly $2$. We additionally proved that for the roommates problem with unweighted agents and strict preferences the popular dimension is $2$. However, for the marriage and roommates problems, adding weights to the agents or ties in the preference lists leaves a gap as we have a lower bound of 2 on the popular dimension and an upper bound of 3. 
    It is an open problem to close the gap between these lower and upper bounds. More specifically, for the marriage problem, we conjecture that the popular dimension is exactly $2$, even with weighted agents and ties in the preference lists.     

\section{Acknowledgements}
This research was supported by NSERC Discovery Grant 2022-04191. The fourth author was also supported by a McCall MacBain Scholarship and partially by FRQNT Grant 332481.
The last author is grateful to Kavitha Telikepalli for introducing him to the topic of popular matchings.


\bibliography{references}

\end{document}